\documentclass[authoryear,12pt,3p,letter]{jowarticle}
\usepackage{graphicx}
\usepackage{amsthm,amsmath}
\usepackage{amssymb}
\usepackage{amsfonts}
\usepackage{natbib}
\usepackage{setspace}
\usepackage[all]{xy}
\usepackage{enumitem}
\usepackage{titlesec}
\usepackage{mathrsfs}

\usepackage{epigraph}
\makeatletter
\renewcommand\section{\@startsection{section}{1}{\z@}{-3.25ex plus -1ex minus -.2ex}{1.5ex plus .2ex}{\normalsize\bf}}
\renewcommand\subsection{\@startsection{subsection}{2}{\z@}{-3.25ex plus -1ex minus -.2ex}{1.5ex plus .2ex}{\normalsize\bf}}
\renewcommand\subsubsection{\@startsection{subsubsection}{3}{\z@}{-3.25ex plus -1ex minus -.2ex}{1.5ex plus .2ex}{\normalsize\bf}}
\makeatother

\newtheorem{innercustomgeneric}{\customgenericname}
\providecommand{\customgenericname}{}
\newcommand{\newcustomtheorem}[2]{%
  \newenvironment{#1}[1]
  {%
   \renewcommand\customgenericname{#2}%
   \renewcommand\theinnercustomgeneric{##1}%
   \innercustomgeneric
  }
  {\endinnercustomgeneric}
}

\newtheorem{thm}{Theorem}
\newtheorem{cor}[thm]{Corollary}

\newtheorem{rem}[thm]{Remark}

\newcustomtheorem{prop*}{Proposition}
\newcustomtheorem{cor*}{Corollary}

\newcommand{\Reals}{\mathbb {R}}

\begin{document}
\begin{frontmatter}
\title{The Local Validity of Special Relativity, \\ Part 1: Geometry}
\author{Samuel C. Fletcher}\ead{scfletch@umn.edu}
\address{Department of Philosophy\\ University of Minnesota, Twin Cities}
\author{James Owen Weatherall}\ead{weatherj@uci.edu}
\address{Department of Logic and Philosophy of Science\\ University of California, Irvine}
\begin{abstract}
In this two-part essay, we distinguish several senses in which general relativity has been regarded as ``locally special relativistic''.
Here, in Part 1, we focus on senses in which a relativistic spacetime has been said to be ``locally (approximately) Minkowskian''.
After critiquing several proposals in the literature, we present a result capturing a substantive sense in which every relativistic spacetime is locally approximately Minkowskian.
We then show that Minkowski spacetime is not distinguished in this result: every relativistic spacetime is locally approximately every other spacetime in the same sense.
In Part 2, we will consider ``locally specially relativistic'' matter theories.
\end{abstract}
\end{frontmatter}

\setlength{\epigraphwidth}{\textwidth}
\epigraph{The general theory of relativity rests entirely on the premise that each infinitesimal line element of the spacetime manifold physically behaves like the four-dimensional manifold of the special theory of relativity.
Thus, there are infinitesimal coordinate systems (inertial systems) with the help of which the $ds$ are to be defined exactly like in the special theory of relativity.
The general theory of relativity stands or falls with this interpretation of $ds$.
It depends on the latter just as much as Gauss’ infinitesimal geometry of surfaces depends on the premise that an infinitesimal surface element behaves metrically like a flat surface element \ldots}
{Albert Einstein to Paul Painlev\'e, December 7, 1921\\
As translated in \citet{LehmkuhlEP}\\
\citep[Doc. 314]{einstein2009collected}}

\doublespacing
\section{Introduction}

\noindent The literature on the foundations of general relativity is replete with claims that, locally, general relativity is like special relativity.
Such claims can take different forms.
Sometimes it is said that, according to general relativity, spacetime is ``locally (approximately) flat" or ``locally Minkowskian," where Minkowski spacetime is the flat, gravitation-free setting of special relativity.\footnote{
    In what follows, a \emph{relativistic spacetime} is a pair $(M,g_{ab})$, where $M$ is a smooth, four-dimensional manifold that we assume to be connected, Hausdorff, and paracompact; and $g_{ab}$ is a smooth, Lorentz-signature metric on $M$. Relativistic spacetimes are the models, or ``solutions'', of general relativity; they represent possible universes, according to the theory. For more on the conventions we adopt here, including the abstract index notation, see \citet{Wald} or \citet{malament2012}.  (Observe, though, that these texts differ in the sign of the metric signature; that choice will not matter for our purposes.)  In this context, Minkowski spacetime is a relativistic spacetime where $M$ is diffeomorphic to $\mathbb{R}^4$ and the metric $g_{ab}$ is flat and geodesically complete.}
In other cases---not necessarily independent of the former ones---the key idea is that matter in general relativity behaves locally ``as if'' it is in the flat-spacetime setting of special relativity.

This locally flat, or locally special relativistic, character of general relativity has been taken to have great  significance. For some authors, it is a crucial heuristic, motivating why one might adopt or postulate the structure and laws of general relativity as a theory of gravitation in the first place \citep{schild1967lectures,ehlers1973survey}.
In this respect it functions similarly to ``correspondence principles'' in the formulation of the old quantum theory \citep{sep-bohr-correspondence}.
For others, claims about local flatness are presented as deductive consequences of general relativity that establish the conditions under which certain general relativistic descriptions of phenomena may be locally well-approximated by special relativistic descriptions \citep{Reichenbach1958,Born1962,ehlers1973survey,torretti1996relativity}; understood in this way, local flatness, if and when it obtains, may provide a sense in which general relativity reduces to, or explains the successful application of, special relativity \citep{nickles1973two,butterfield2011emergence}.
And for still others, the locally special relativistic character of the theory is invoked to support a privileged role for special relativity in interpreting general relativity \citep{ehlers1973survey,MTW1973,Brown1997,Brown2005,Knox2013-KNOESG}.
For many commentators, local flatness is intimately connected with other principles that they take to be foundational to or important in general relativity, such as some version of the equivalence principle  \citep{schild1967lectures,ehlers1973survey,Knox2013-KNOESG,Brown2005,ReadBrownLehmkuhl2018}.

But despite the ubiquity of these claims, there is little clarity or agreement within the literature concerning what, precisely, such assertions are supposed to mean.
Our goal here and in the sequel to this paper is to offer a new perspective on the sense, or senses, in which general relativity is locally like special relativity, in the service both of clarifying the sense in which it is true that spacetime is (approximately) locally flat and in assessing what significance that has for the local dynamics of matter.\footnote{
    Given that the project here is to make sense of claims about local flatness in general relativity, and given that local flatness is implicated in some formulations of the equivalence principle, one might take the present project to be part of a long tradition of work attempting to precisely recover what various authors have meant by the equivalence principle or some other alleged principle, such as ``substantive general covariance'' \citep[for which see, e.g.][]{NortonEP,NortonGC,PooleyGC,LehmkuhlEP}.
    But we see our project differently.
    In particular, we do not seek to trace the historical development of claims about local flatness, nor to adjudicate historical debates from a contemporary perspective.
    Instead, we seek to isolate a sense in which spacetime is locally approximately flat in general relativity and to discuss its significance.
    Our critical remarks on other proposed explications of the claim are offered to clear the ground; we take the arguments given here to be of interest irrespective of whether the claims we defend align with those of others in the literature.
    We also see this project as a continuation of the research programs sketched in, for instance, \citet{WeatherallDogmas}, to isolate precise mathematical statements that might serve as sufficient conditions for theorems concerning when matter theories are adapted to a certain geometry, or, relatedly, the program begun in \citet{FletcherInvitation,Fletcher2020} to better explicate the relationship between general and special relativity.}
The present paper focuses on geometrical aspects of the question, with an emphasis on local approximate flatness; the sequel, which will make use of the results here, will consider several senses in which matter dynamics may be locally special relativistic.

We will begin, in section \ref{sec:wrong}, by presenting several possible interpretations---or perhaps better, explications---of the assertion ``spacetime is locally (approximately) flat'', all inspired by attempts in the literature to state the claim precisely.  
As we will argue, each of these is inadequate---either because it is false, misleading, or does not perspicuously capture the relevant facts.
Still, we claim there is a precise sense in which every relativistic spacetime \emph{is} locally approximately flat.
In section \ref{sec:right}, we will introduce that sense, which is common folklore in mathematical relativity but rarely, to our knowledge, stated precisely or proved, at least in its full generality.\footnote{
    There are partial exceptions.  For instance, \citet[pp.~20, 44]{ehlers1973survey} calls versions of this Theorem a ``well known theorem of differential geometry'', while \citet{poisson2004relativist} and \citet{poisson2011motion} call a somewhat weaker result the ``local flatness theorem'' and prove it. But neither treatment is as general as it could be, in ways that may obscure its significance.  Even so, the result is not original: it is a trivial consequence of work by \citet{oRaifeartaigh1957fermi}, and it is invoked by others such as \citet{Geroch+Jang} and \citet{Geroch+Weatherall}.}
This will be expressed by Theorem \ref{thm:normal}.
We will discuss some advantages of this approach, which include that it clarifies the sense in which ``local flatness structure'' fails to be unique.

In section \ref{sec:triviality}, we will argue that once it is clear what local approximate flatness amounts to, there are reasons to be cautious about attributing too much significance to it.
In particular, we will argue, there is nothing special about \emph{flatness} in Theorem \ref{thm:normal}.
In fact, \emph{every} relativistic spacetime locally approximates every other relativistic spacetime in the same sense that Minkowski spacetime does.
In other words, while it is true that spacetime is locally approximately Minkowskian, so, too, is it locally approximately (anti-)de Sitterian, Schwarzschildian, Kerrist, G\"odelian, and so on.
This will be our Theorem \ref{thm:universal}.
The upshot is that local approximate flatness, as a feature of pseudo-Riemannian manifolds, might be better characterized as a \textit{universal approximation property}: to first order (in derivatives), all metrics of a given signature locally approximate one another.
It is only at \emph{second} order---that is, the order of curvature---that these metrics fail to approximate one another, even at a point.
This fact is arguably deep, and closely connected to the fact that curvature can be represented as a tensor.
But in our view, it is not naturally expressed as the claim that ``spacetime is (approximately) locally flat''---even though it happens to imply that spacetime \emph{is} locally approximately flat (among other things).

In section \ref{sec:symmetry}, we will show how Theorems \ref{thm:normal} and \ref{thm:universal}  offer additional insight into another claim closely related to the claim that spacetime is locally approximately Minkowskian, which is that relativistic spacetimes admit local approximate Poincar\'e symmetries \citep{ReadBrownLehmkuhl2018, Fletcher2020}.
Finally, we will offer some brief concluding remarks in section \ref{sec:conclusion}.
In Part II, we will turn to the relationship between local approximate flatness and the behavior of matter.

\section{What local approximate flatness is not}
\label{sec:wrong}

We begin by clearing the air.  What should one \emph{not} mean by the assertion ``spacetime is locally flat'' in general relativity?\footnote{An anonymous reviewer questions our use of the word ``interpretation'' throughout this section.  Here is how we see what we are doing.  The claim ``spacetime is locally (approximately) flat'' appears to be ambiguous, in the sense that different authors interpret it differently.  Here we identify several such interpretations and offer precise statements (``explications'') intended to capture the meaning of the claim under each interpretation.

   Note, too: Occasionally the term ``locally flat'' is used in geometric topology \citep[e.g.,][]{brown1962locally} to denote a particularly ``nice'' or ``neat'' embedding of one topological manifold into another.
    Clearly that usage is not applicable to the case at hand.}

As a first pass, recall that a relativistic spacetime $(M,g_{ab})$ is \textit{flat} just when its Riemann tensor, $R^a{}_{bcd}$ vanishes everywhere.
Recall further that $(M,g_{ab})$ is \textit{locally isometric to} a spacetime $(M',g'_{ab})$ when, for each point $p \in M$, there exists a neighborhood $U_p$ and a smooth map $\psi_p: U_p \to M'$ such that $(\psi_p)^*(g'_{ab}) = g_{ab}$ at each point of $U_p$.\footnote{
    Note that in the literature, ``locally isometric'' denotes several distinct relations.
    For example, in contrast with this asymmetric relation, one can also define its symmetrization: spacetimes $(M,g_{ab})$ and $(M',g'_{ab})$ are \textit{(mutually) locally isometric} when each is locally isometric to the other.}
This condition captures a sense in which ``locally'' the spacetime $(M,g_{ab})$ is equivalent to, or has the same structure as, (some region or other of) $(M',g'_{ab})$, even though globally the two spacetimes could be completely different.

These definitions suggest a natural, literal interpretation of the claim that (any) spacetime is locally flat or locally Minkowskian.
\begin{description}
    \item[Literal Interpretation:] Every relativistic spacetime is locally isometric to a flat spacetime (e.g., a region of Minkowski spacetime).
\end{description}
Unfortunately, interpreted in this way, the claim is simply false: not every relativistic spacetime is locally flat in this sense.
And it would not help to restrict attention only to those spacetimes that \emph{are} locally flat in the sense, such as by claiming that it is only those spacetimes that are physically reasonable \citep[cf.][]{Manchak}.
Indeed: a relativistic spacetime is locally isometric to flat spacetime if and only if it is flat simpliciter, because Riemann curvature is preserved under isometry.
So if a spacetime is locally isometric to a flat spacetime, its Riemann curvature must vanish at every spacetime point.
Einstein's equation, meanwhile, implies that curvature is generally non-zero in the presence of matter.

Is this first interpretation ever endorsed in the literature?
Perhaps not in such an explicit, and obviously unacceptable, form.
But it is arguably a mere rephrasing of another interpretation that has been widely endorsed.
On this interpretation, spacetime is locally flat in the sense that one can always ``transform away'' arbitrary gravitational effects by choosing appropriate coordinates, much as one can fictitious forces.
Such interpretations identify gravitational effects not directly with curvature, but with the coordinate-dependent Christoffel symbols of the metric connection.
This claim was what many early commentators, such as \citet[pp.~705--6]{Pauli}, identified with the equivalence principle.\footnote{
    See \citet{NortonEP} or \citet{LehmkuhlEP} for a discussion of this point, and for a contrast with Einstein's own views.
    In a word, Einstein only claimed equivalence for a homogeneous gravitational field, i.e., constant Christoffel symbols \citep{JanssenRelative}.
    One can show that the Christoffel symbols are constant in a neighborhood iff the metric is flat there, so Einstein's view is \textit{not} an example of the following first pass of the Coordinate Chart Interpretation.}
\begin{description}
    \item[Coordinate Chart Interpretation, first pass:] In any sufficiently small region of any relativistic spacetime, coordinates may be chosen relative to which the Christoffel symbols of the Levi-Civita (i.e., unique torsion-free metric-compatible) connection vanish.
\end{description}
But once again, as has been observed by many others \citep[e.g.,][]{Eddington,Synge,Friedman1983,NortonEP}, this claim is false in general; and it is true of a spacetime precisely when that spacetime is flat.\footnote{
    The fact that this claim was refuted by Eddington in the 1920s did little to stop others from repeating it from time to time. 
    Even \citet[p.~285]{MTW1973} seem to endorse this claim, for instance when they assert that ``one can always construct local inertial frames at a given event $\mathcal{P}_0$; and as viewed in such frames, free particles must move along straight lines, at least locally---which means $\Gamma^\alpha_{~\beta\gamma}$ must vanish, at least locally.''
    Here, the $\Gamma^\alpha_{~\beta\gamma}$ are the Christoffel symbols for the coordinate system generated by the mentioned frame.
    Similarly, some pages later they write: ``In every local region there exists a local frame (``freely falling frame'') in which all geodesics appear straight (all $\Gamma^\alpha_{~\beta\gamma} = 0)$'' \citep[p.~297]{MTW1973}.
    Now, Misner, Thorne, and Wheeler cannot truly mean to endorse this claim---and elsewhere in their book, they are more careful.
    Nonetheless, there is some value in emphasizing that this is false, given how frequently it has been repeated!}
Here the basic facts are that the Christoffel symbols are constant in an open neighborhood iff the Riemann tensor vanishes in that neighborhood, and that the Riemann tensor is a tensor---and thus it vanishes in any coordinate system iff it is the zero tensor.

The problem with these first two readings is that they insist that spacetime is flat ``in a neighborhood'' of any point, which can hold only if it is flat everywhere.
But sometimes, authors who appear to endorse such readings explain that the neighborhoods in question must be ``infinitesimal'' \citep[see, e.g.][p.~226]{Reichenbach1958}.
This suggests that perhaps local flatness should not be associated with an open set of spacetime points at all, but rather with the points themselves.
One possible reading of this idea would be that ``local flatness'' claims concern the structure of the tangent space at each point, since this space can be thought of as a representation of the linearized, or first-order, structure of an infinitesimal neighborhood of the point.
\citet[p.~71]{Brown1997} makes this link explicitly, writing that ``relative to local inertial frames (defined in the infinitesimal neighborhood of any event) all the laws of physics take on their special relativistic form.
Put another way, the tangent space structure in GR is everywhere `Lorentzian'.''\footnote{
    See also \citet[pp.~183--4]{Friedman1983}, who similarly draws a connection between Minkowski spacetime and the tangent space at each point of a relativistic spacetime.}

These considerations lead to our third interpretation:
\begin{description}
    \item[Tangent Space Interpretation:] The tangent space at a point of spacetime is, or is equivalent to, Minkowski spacetime.
\end{description}
At first blush, this proposal has something going for it.
The tangent space at any point of a spacetime manifold is a four dimensional vector space, which means, in particular, that it carries the structure of the smooth manifold $\mathbb{R}^4$, just as Minkowski spacetime does.
Moreover, the spacetime metric induces a Lorentz-signature metric on the tangent space, and so there is a sense in which the tangent space metrical structure is also arguably the same as that of Minkowski spacetime.

Nonetheless, there are in fact important differences between the tangent space of a relativistic spacetime  and Minkowski spacetime.
In the first place, the tangent space is a vector space, while Minkowski spacetime has the structure of an affine space.
The difference is significant, as the lack of a preferred point in the latter---the zero element in the former---precludes the classification of individual points as being spacelike, timelike, or null (as opposed to classifying pairs of points as ``spacelike [etc.] related'').
The difference also bears on the physical interpretations of the spaces.
The points of Minkowski spacetime represents spatiotemporal events, while a tangent space represents instantaneous directions of curves at one of those events.

An advocate for this interpretation might reply that ``local flatness'' means that infinitesimal neighborhoods of each point---that is, the tangent space---should be thought of as equivalent to Minkowski spacetime \emph{with a distinguished point}, since, after all, we are representing a neighborhood of a particular point.
Alternatively, one might argue that vector space structure is \emph{more} structure than affine space structure, and so if any point can be be associated with a vector space, then \emph{ipso facto} it can be associated with an affine space.
Fine.
But even if we set aside the structural differences between Minkowski spacetime and the tangent space at a point of a relativistic spacetime, if the tangent space interpretation is all that is meant by ``local flatness'', it is strikingly weak.
This is because Riemann curvature is a tensor field, and so it determines a tensor acting on the tangent space at each point.
Thus, even from the perspective of the tangent space at a point, one can ``see'' the curvature of spacetime near that point by considering the curvature tensor there.

More generally, curvature is a measure of the failure of parallel transport of vectors and tensors around (infinitesimally small!) closed curves to return a vector or tensor to its original value; in that sense, it is a characterization of the relationship between the tangent spaces at nearby points, as determined relative to some connection.
Observing that the tangent space at each point has the structure of a vector space with a Minkowskian inner product, though, says nothing at all about parallel transport in small neighborhoods of the point.
And so the claim that the tangent space is ``flat'' loses any relation to the meaning of curvature on a manifold in the first place.
At best, on this interpretation the claim that spacetime is locally flat amounts to the observation that points of spacetime can be associated with other spaces that are, in some sense, flat---without capturing any sense in which that flatness reflects or approximates the local curvature structure of the manifold.

\citet[p.~240]{torretti1996relativity} proposes a different take on the tangent space interpretation.
He writes, ``The Minkowski inner product on each tangent space induces---through the exponential mapping---a local approximate Minkowski geometry on a small neighborhood of each worldpoint.''
Fix a spacetime $(M,g_{ab})$.
The exponential map, at a point $p\in M$, is a map from a neighborhood $O_p$ of the $0$ vector in the tangent space $T_p M$ to some open set $U_p\subseteq M$ containing $p$; the map is defined relative to a derivative operator on $M$ by taking each vector $\xi^a\in O_p$ to the point $\gamma(1)$, where $\gamma$ is the unique geodesic through $p$ with tangent $\xi^a$ at $p$.
(Here $O_p\subset T_pM$ is chosen so that the exponential map is injective, which is always possible.)
The exponential map generates a diffeomorphism between $O_p$ (conceived as a manifold) and $U_p$, and it generates coordinates on that latter set with certain nice properties.
In particular: they are \emph{normal} coordinates at $p$, which means that the Christoffel symbols of the Levi-Civita derivative operator in those coodinates vanish at $p$, or in other words, the Levi-Civita derivative operator and the coordinate derivative operator agree there.
They are also \emph{Lorentz} coordinates, which means that the metric in those coordinates has the same coordinate representation at $p$ as the Minkowski metric in standard coordinates---that is, as the matrix $diag(1,-1,-1,-1)$.
Thus, there is a certain sense in which these coordinates are adapted to the metric and derivative operator at $p$---and since they are coordinates, and coordinate derivative operators are always flat, they can be thought of as generating a ``local Minkowskian'' structure on $U_p$ that agrees with the background spacetime structure at $p$ and approximates it elsewhere on $U_p$.

The fact that this local Minkowski geometry is only \emph{approximate} importantly distinguishes Torretti's claims from the literal interpretation discussed above.
As Torretti himself emphasizes, the existence of normal coordinates generated in this way in no sense implies that the spacetime is flat, even at $p$.
As he acknowledges, ``The mere fact that the tangent space has a Minkowskian \ldots inner product---as it obviously does everywhere, by definition, on the manifolds under consideration---says nothing whatsoever about the value of the Riemann tensor and the manifold's departure from flatness at that point'' \citep[p.~314n13]{torretti1996relativity}.

Torretti's invocation of Lorentz normal coordinates brings us to another common interpretation of the local flatness claim.
On this interpretation, it is the existence of certain normal coordinates at each point or along certain curves, such as timelike geodesics, that is supposed to capture the sense in what spacetime is locally flat.
\begin{description}
    \item[Coordinate Chart Interpretation, second pass:] At any point of any relativistic spacetime (or along certain curves), local coordinates may be chosen so that, at that point (or along that curve), (a) the components of the metric agree with the Minkowski metric in standard coordinates and (b) all Christoffel symbols vanish.
\end{description}
This is true.
And as we will discuss in the next section, it is very close to our own preferred interpretation.
But even so, we think this way of stating things obscures what is going on, for several reasons.

First, it is not clear what coordinates have to do with the basic claim of local flatness.
On the one hand, \emph{any} coordinate system gives rise to a flat derivative operator, and coordinates can always be used to define flat metrics.
If local flatness is nothing more than the observation that there exist coordinates in neighborhoods of any points, then, just as with the tangent space interpretation, the present interpretation seems too weak to be of interest.
In particular, it seems to say nothing about the local curvature at a point or nearby.
Now, it is true that on this interpretation one invokes \emph{special} coordinates---viz., Lorentz normal ones---but the significance of those coordinates requires further commentary.
What does the ``form'' of the metric in some coordinate system tell us about the metric or its (covariant) derivatives, all of which are coordinate independent structures?  What does it tell us about local curvature?

Likewise, what significance should be attributed to the fact that Christoffel symbols can be made to vanish at a point or along a curve?
After all, given any derivative operator, including any flat derivative operator, one can \emph{also} always find coordinates for which Christoffel symbols for that derivative operator do \emph{not} vanish at any point.
So the existence of coordinate systems in which Christoffel symbols do (or do not) vanish does not obviously reveal any coordinate-independent facts about the derivative operator.
(As we will see, there are such facts lurking in the background here; our point is that, without further discussion, the interpretation as stated does not seem like a perspicuous way of expressing them.)

A second set of issues concerns whether the interpretation is intended as an assertion of the existence of coordinates with certain properties (such as being normal), or if it is supposed to come with a further claim about the significance of particular coordinate systems---say, ones constructed from the exponential map, \`a la Torretti---in which case, it is not clear just which properties are supposed to be the ones that realize the claim about local flatness.
That is: is the relevant fact supposed to be the existence of \emph{normal} coordinates (or Lorentz normal coordinates)---or special classes of Lorentz normal coordinates generated via specific construction procedures?
This ambiguity also leads to confusion about whether it is Lorentz normal coordinates defined in a neighborhood of a point, ones defined in neighborhoods of certain curves, such as timelike geodesics, or perhaps a more general class of normal coordinates that are supposed to be the relevant ones.
If it is not clear what features of these coordinates are supposed to be salient, it is hopeless to try to establish the general existence or uniqueness conditions for such coordinates.

To make matters worse, some authors move quickly from the observation (or argument) that certain coordinates exist to comments about their physical significance.
For instance, some authors suggest that local flatness obtains because one can find a class of coordinates known as Fermi normal coordinates along timelike geodesics, which are constructed by parallel transporting along the geodesic an orthonormal frame whose timelike vector is tangent to the geodesic, and then extending it to a neighborhood of the curve by a construction analogous to that described for the exponential map; then they immediately go on to argue that Fermi normal coordinates are analogous to ``inertial frames'' in special relativity.\footnote{
    See, for instance, \citet[pp.~20--23]{schild1967lectures}, \citet[pp.~199--200]{Friedman1983}, or \citet[\S2]{Knox2013-KNOESG}.  \citet[pp.~11--12]{poisson2004relativist} offers a nice example of the ambiguity: he states a ``local flatness theorem'' that asserts the existence of Lorentz normal coordinates; he then proceeds to indicate that the particular coordinates he constructs to prove the theorem indicate something about what freely falling observers will ``see''.
    But Lorentz normal coordinates are not unique, and so it is unclear whether local flatness is an assertion about the existence of such coordinates in the first place, or one about the further interpretive significance of a special class of such coordinates, arrived at through Poisson's construction.}
The result is that it is unclear whether local flatness is meant to be a claim about the existence of certain coordinate systems, which in turn expresses something about the local geometry of relativistic spacetimes, or if it is supposed to include a further interpretive claim about idealized measurement apparatuses of natural motion, which, of course, would go far beyond any facts about curvature, local or otherwise.

Still, as we said above, we think this final interpretation does express something with meaningful content about the structure of relativistic spacetimes---something that is well-expressed by the claim ``spacetime is locally approximately flat''.
In the next section, we will restate and generalize this interpretation in a way that makes that content more perspicuous.

\section{In what sense is spacetime locally flat?}\label{sec:right}

We have now presented four interpretations of the claim ``spacetime is locally flat in general relativity''.
Two of these were unacceptable because they were simply false claims about relativistic spacetimes; the third was unacceptable because it had so little content that it seemed it could not do any foundational work at all.
The final one, related to the existence of normal coordinates on certain neighborhoods, does capture a sense in which spacetime is locally approximately flat---but we argued that common expressions of this interpretation in the literature are unsatisfactory.
We will now proceed to rephrase and generalize the final interpretation of the previous section.\footnote{\label{wallace}
    As will be clear presently, the relationship between the sense of ``local approximate flatness'' captured by Theorem \ref{thm:normal} and that expressed by the second coordinate chart interpretation above is very similar to that between ``intrinsic'' or coordinate-free characterizations of geometrical structures and ones that invoke classes of coordinate systems adapted to those structures \citep{Wallace}.
    To some extent, any preference between them is a matter of taste, and for many purposes, it is very useful to have \emph{both} characterizations available.
    But as we discuss below, reframing things as we do in this section permits one to capture the sense in which local flatness structure, though it always exists, is not unique---something that has been unclear in other discussions.}

We begin by stating and proving a theorem.\footnote{
    We do not claim that this theorem is notably original---it follows trivially from work by \citet{oRaifeartaigh1957fermi}.
    (See also \citet{iliev2006handbook}.)
    But we have not seen it stated in this form before, nor have we seen it discussed in the context of claims about local flatness in the philosophical literature.
    So we think there is some value in stating and proving it here.}
\begin{thm}[Local Flatness]\label{thm:normal} Given any spacetime $(M,g_{ab})$, any embedded curve $\gamma:I\rightarrow M$ therein, and any point $p\in \gamma[I]$, there exists, on some neighborhood $O$ containing $p$, a flat metric $\bar{g}_{ab}$ such that on $\gamma[I]\cap O$, (a) $g_{ab}=\bar{g}_{ab}$ and (b) $\nabla=\bar{\nabla}$, where $\nabla$ and $\bar{\nabla}$ are the Levi-Civita derivative operators associated with $g_{ab}$ and $\bar{g}_{ab}$, respectively.
\end{thm}

Here an ``embedded curve'' is a curve whose image is a one-dimensional embedded submanifold.
\begin{proof}
Given any torsion-free derivative operator $\hat{\nabla}$ on a smooth manifold $M$, there exists, in sufficiently small neighborhoods of sufficiently small segments of an arbitrary (non-self-intersecting) curve $\gamma$, flat derivative operators that agree with $\hat{\nabla}$ on the intersection of the neighborhood and segment of $\gamma$ \citep[Thm.~II.3.2]{iliev2006handbook}.
Let $\nabla$ be the Levi-Civita derivative operator compatible with $g_{ab}$.
For any point $p\in \gamma[I]$, choose some such flat derivative operator $\bar{\nabla}$ defined on a suitably small neighborhood $O$ meeting $\gamma[I]$ and agreeing with $\nabla$ on $\gamma[I]\cap O$.
(That $\gamma[I]$ is embedded implies that $O$ can be chosen so it intersects only the desired segment of $\gamma$.)
Pick, at $p$, an orthonormal frame $\{\overset{i}{u}{}^a\}_{i\in\{0,\ldots,3\}}$, with $\overset{0}{u}{}^a$ timelike, and extend it, by parallel transport using $\bar{\nabla}$, to all of $O$.
Define $\bar{g}_{ab}=\overset{0}{u}_a\overset{0}{u}_b - \sum_{i=1}^3 \overset{i}{u}_a\overset{i}{u}_b$ on $O$.
Then $\bar{g}_{ab}$ is flat by construction and $\bar{\nabla}$ is its Levi-Civita derivative operator.
Since $\gamma^c\nabla_c g_{ab} = \mathbf{0} = \gamma^c\bar{\nabla}_c\bar{g}_{ab} = \gamma^c\nabla_c\bar{g}_{ab}$ on $\gamma[I]\cap O$, where $\gamma^a$ is tangent to $\gamma$, $\gamma^c\nabla_c(g_{ab}-\bar{g}_{ab}) = \mathbf{0}$ on $\gamma[I]\cap O$, i.e., $g_{ab}-\bar{g}_{ab}$ is constant on $\gamma[I]\cap O$ (with respect to $\nabla$).
Since $g_{ab}-\bar{g}_{ab} = \mathbf{0}$ at $p$, it follows that $\bar{g}_{ab}$ agrees with $g_{ab}$ everywhere on $\gamma[I]\cap O$.
\end{proof}

There are several technical remarks to make before proceeding.  First, note that the theorem does not require $\gamma$ to be a geodesic or even timelike. Indeed, it can be generalized from embedded curves to embedded submanifolds with vanishing intrinsic curvature \citep[Thm.~II.5.1, II.5.2]{oRaifeartaigh1957fermi,iliev2006handbook}.
This shows that the interpretation of $\gamma$ as the worldline of an observer is not essential to the result: all that is needed is an intrinsically flat embedded submanifold, of which points and embedded curves are always examples.  Of course, the fact that such metrics exist for general curves specializes to the case of timelike geodesics.  Second, as we have set things up, the flat metrics exist only along segments of the image of a curve; in certain special cases, they can be extended to the entire curve.\footnote{\citet[\S II.3.1]{iliev2006handbook} has a nice discussion of this point.}  But for present purposes, we prefer to emphasize the more general, and also more local, claim, since after all it is \emph{local} approximate flatness that is at issue.  And finally, we note that our statement assumes embedded curves, and works in neighborhood of points of the image of the curve; one could relax the assumption that the curve is embedded, but then one would have to work with neighborhoods (in $M$) of images of neighborhoods of parameter values (in $I$), which seems less natural to us.

As we have indicated before, there is a certain sense in which Theorem \ref{thm:normal} expresses the same facts as a (strengthened) version of the second coordinate chart interpretation above.  In particular, anything that follows from the existence of normal coordinates also follows from the existence of the flat derivative operators considered here, and vice versa.  This is because normal coordinates always give rise to a flat derivative operator that will agree with the spacetime Levi-Civita derivative operator wherever Christoffel symbols vanish; and parallel transporting the spacetime metric off of the curve using this coordinate derivative operator will give rise to a flat metric that agrees with the spacetime metric on the curve.  Conversely, given a flat metric with the properties described in the theorem, one can always find an isometry from the region where the metric is defined to a region of Minkowski spacetime, and then use that isometry to pull back standard Minkowski coordinates; coordinates constructed in this way will automatically be normal.

Even so, we suggest that the most natural interpretation of the claim that spacetime is locally approximately flat is given by Theorem \ref{thm:normal}.  Why?  First, we have stated this result as a claim about the existence of certain structures on regions of spacetime---namely, a flat metric and derivative operator that coincide, along curve segments, with the spacetime metric and its derivative operator.  Moreover, this flat metric approximates the background metric near $p$, in a sense we can make precise.\footnote{
    We adapt the following definitions from \citet{Fletcher2020}, who treats the special case of approximate local spacetime symmetries, and the interpretation of the structures invoked from \citet[p.~20]{fletcher2018} and \citet[\S4]{fletcher2019reduction}.}
Fix a spacetime $(M,g_{ab})$ and any open set $U\subseteq M$ with compact closure.\footnote{
    Nothing in the definition demands specializing to the case of $U$ being relatively compact, but this restriction is most relevant for what follows.
    Implicit in this choice is understanding that for the present investigation, approximation at single spacetime points is insufficient but approximation across the entire spacetime is unnecessary.
    What is important seems to be approximation on extended but bounded regions.
    That is why we examine relatively compact regions, similarity across which can be captured with the compact-open topologies on spacetimes \citep{fletcher2016similarity}. }
Choose any smooth Riemannian metric $h_{ab}$ on $U$.
Physically, any such metric can be determined by a smooth, orthonormal frame field $\{\overset{i}{u}{}^a\}_{i\in\{0,\ldots,3\}}$ on $U$, as $\sum_{i=0}^3 \overset{i}{u}_a\overset{i}{u}_b$ is a smooth Riemmanian metric.
(It could instead be determined by a coordinate chart on $U$ in an analogous way; which one chooses is not essential for what follows.)
We may then define, relative to $h^{ab}$, a norm on covariant tensors $f_{a_1\cdots a_n}$ at a point by:\footnote{
    This definition can be extended to arbitrary tensors, but for present purposes only covariant ones are of interest, and so we limit attention to those to simplify notation.}
\[|f|_h = |h^{a_1 b_1} \cdots h^{a_n b_n} f_{a_1 \ldots a_n} f_{b_1 \ldots b_n} |^{1/2}.\]
(This norm is the Frobenius norm on the tensors expressed in terms of their components relative to the frame field.)
Using this family of norms, we can define a family of distance functions on tensors as:
\[d_U(f,f';h,k) = \max_{j \in \{0, \ldots, k\}} \sup_U |(\nabla)^j (f-f')|_h,\]
where $(\nabla)^{j}$ abbreviates ``act with $j$ derivatives'', $\nabla$ is the Levi-Civita derivative operator determined by $h_{ab}$, and $(f-f')$ abbreviates $f_{a_1\cdots a_n}-f'_{a_1\cdots a_n}$.
What this distance function does is return the greatest distance, relative to $h_{ab}$, between $f$ and $f'$ or any of their first $k$ derivatives, ranging over all points in $U$.
And the distance between $f$ and $f'$ (or their derivatives, respectively) is simply the Euclidean magnitude of the differences in their components expressed with respect to the frame field $\{\overset{i}{u}{}^a\}_{i\in\{0,\ldots,3\}}$.

The case of greatest interest here will be when we use distance functions defined in this way to measure distances between different Lorentzian metrics on $U$.
Indeed, let $g_{ab}$, $\bar{g}_{ab}$, and $O$ be as in the statement of Theorem \ref{thm:normal}.
Then it immediately follows from the smoothness of $g_{ab}$ and $\bar{g}_{ab}$ that for any $h_{ab}$ on $O$ and any $\epsilon > 0$, there exists a neighborhood $U\subseteq O$ such that $d_U(g,\bar{g};h,1) < \epsilon$.
Thus we see that not only do the two metrics coincide at $p$, but they also approximate one another, to first order, arbitrarily well in sufficiently small neighborhoods of $p$.
Setting things up this way makes clear the sense in which there is a flat \emph{spacetime structure} that approximates an arbitrary metric and derivative operator near a point or (arbitrary) curve (as opposed to admitting a representation, such as $diag(1,-1,-1,-1)$, in which it has a particular syntactic form), which we take to be an assertion of ``approximate flatness'' with more clearly defined implications.
In other words, Theorem \ref{thm:normal} captures the sense in which every relativistic spacetime can be approximated locally by a flat spacetime.

As we saw above, claims about local flatness are sometimes expressed as the assertion that (every) spacetime is locally (approximately) Minkowskian.  This idea can be captured in the present context using Theorem \ref{thm:normal} and the notion of ``approximate isometry'' introduced by \citet{Fletcher2020}.  With this distance function in hand, consider spacetimes $(M,g_{ab})$ and $(M',g'_{ab})$ with open subsets $U \subseteq M$ and $U' \subseteq M'$, respectively, both of compact closure, and suppose that there is a diffeomorphism $\chi : U' \to U$.
Then given any Riemannian metric $h_{ab}$ on $U$ and any integer $k\geq 0$, we say that $\chi$ is a $(h,k,\epsilon)$-isometry between $(U',g'_{ab})$ and $(U,g_{ab})$ whenever $d_U(g,\chi_*(g');h,k) < \epsilon$.\footnote{Here and in what follows, if $(M,g_{ab})$ is a spacetime and $U$ is an open subset of $M$, we will use ``$(U,g_{ab})$'' to denote the spacetime with manifold $U$ and metric $g_{ab}$ restricted to $U$.  We will not explicitly indicate that $g_{ab}$ is restricted in this way.}  This comports with the fact that $\chi$ is an isometry simpliciter between $(U',g'_{ab})$ and $(U,g_{ab})$ when $d_U(g,\chi_*(g');h,k) = 0$ for all $h$ and $k$.  Such $(h,k,\epsilon)$-isometries can be thought of as ``approximate isometries'', since they capture the idea that the two metrics agree only to a degree of approximation ($\epsilon$) with respect to some reference structure ($h_{ab}$) and only up to some fixed order of differentiation ($k$).

Theorem \ref{thm:normal} then has the following corollary.
\begin{cor}\label{cor:normal}
Given any spacetime $(M,g_{ab})$, embedded curve $\gamma:I\rightarrow M$, point $p\in \gamma[I]$, compact neighborhood $U$ of $p$, Riemannian metric $h$ on $U$, real $\epsilon > 0$, and point $p'$ in Minkowski spacetime $(\Reals^4,\eta_{ab})$, there exist neighborhoods $O \ni p$ and $O' \ni p'$, an embedded curve $\gamma':I'\rightarrow \Reals^4$ with $p' \in \gamma'[I']$, and an $(h,1,\epsilon)$-isometry $\chi: O' \to O$ between $(O,g_{ab})$ and $(O',\eta_{ab})$ satisfying $\chi \circ \gamma' = \gamma$ on $I'$ and $\chi^*(g_{ab})=\eta_{ab}$ on $\gamma'[I']$.
\end{cor}
\begin{proof}
By Theorem \ref{thm:normal}, there is a neighborhood $O \ni p$ on which exists a flat metric $\bar{g}_{ab}$ such that on $\gamma[I]\cap O$, $g_{ab}=\bar{g}_{ab}$ and $\nabla=\bar{\nabla}$, where $\nabla$ and $\bar{\nabla}$ are the Levi-Civita derivative operators associated with $g_{ab}$ and $\bar{g}_{ab}$, respectively. Furthermore, $O$ can be chosen to be relatively compact and sufficiently small to be in $U$, diffeomorphic to $\Reals^4$, and satisfy $d_O(g,\bar{g};h,1) < \epsilon$.
This last property follows from the facts that $h_{ab}$ is smooth and $O$ is a relatively compact neighborhood of $\gamma[I]\cap O$. Therefore there is an isometric embedding $\psi$ of $(O,\bar{g}_{ab})$ into $(\Reals^4,\eta_{ab})$ which, without loss of generality, can be chosen so that $\psi(p)=p'$.
Then define $I' = \gamma^{-1}[\gamma[I] \cap O]$, $\gamma'= \psi \circ \gamma_{|I'}$, $O' = \psi[O]$, and $\chi = \psi^{-1}_{|O'}$.
By definition, $p' \in \gamma'[I']$, $\chi \circ \gamma' = \gamma$ on $I'$, and $\chi^*(g)=\chi^*(\bar{g})=\eta$ on $\gamma'[I']$, hence $d_O(g,\phi_*(\eta);h,1) < \epsilon$.
\end{proof}

This corollary captures the sense in which every spacetime is ``locally approximately Minkowski'', to first order, in neighborhoods of any embedded curve.  One can also state a sense in which this local approximation holds \emph{only} to order 1.  In fact, we have the following:
\begin{rem}
Cor. \ref{cor:normal} holds as stated, but with $k>1$, only if the Riemann tensor associated with $g_{ab}$ vanishes everywhere.
\end{rem}
\noindent Thus we see that even \emph{approximate} local flatness to order 2 or greater can hold only for spacetimes that are flat simplicter.
This result follows from the tensorial character of curvature as the deviation of the second covariant derivative of the metric from zero.
Any metric with non-zero curvature at a point $p$ will fail to fall within $\epsilon$ of a flat metric in any neighborhood of that point, to order 2 or greater, for any fixed $h_{ab}$ and sufficiently small $\epsilon$.

Yet another advantage of the present approach to local approximate flatness is that is clarifies the uniqueness properties of this approximating structure.
As we noted above, in general there are \emph{many} normal coordinates associated with any point or curve; one might wonder whether some of them are privileged, or even how they are related.
Theorem \ref{thm:normal} provides insight into this situation.
First, given a flat metric with the properties described in the Theorem, there will be be many different normal coordinates adapted to that metric, corresponding to different choices of standard Minkowski coordinates.
So we see clearly that if the flat metric is the structure we care about, then it cannot be a particular choice of normal coordinates that is privileged; at best, it is an equivalence class of such coordinates that are all adapted to the same flat metric.

On the other hand, we also have the following.
\begin{cor} In general, for any sufficiently small neighborhood of any point on the image of an embedded curve in a spacetime, there are (infinitely) many flat metrics have the properties described in Theorem \ref{thm:normal}.\footnote{
    Another construction is available that may help drive the point home.
    We will describe it for geodesics, as it's simplest to state in that context; similar constructions are available for general curves.
    Let $\gamma:I\rightarrow M$ be a geodesic in a spacetime $(M,g_{ab})$ with Levi-Civita derivative $\nabla$, and let $\bar{\nabla}$ be a flat derivative operator on a neighborhood $O$ of (some point of) $\gamma[I]$ agreeing with $\nabla$ on $\gamma[I]$.
    Then $\bar{\nabla}'=(\bar{\nabla},\alpha\gamma^a{}\bar{\nabla}_b\alpha \bar{\nabla}_c\alpha)$, where $\gamma^a$ is the unit tangent to $\gamma$ parallel transported to all of $O$ with $\bar{\nabla}$ and $\alpha$ is any scalar field such that (a) $\alpha=0$ on $\gamma[I]\cap O$; (b) $\bar{\nabla}_a\alpha\neq \mathbf{0}$ on $O\cap \gamma[I]$; and (c) $\bar{\nabla}_b\bar{\nabla}_a\alpha = \kappa\bar{\nabla}_b\alpha\bar{\nabla}_a\alpha$ for some smooth scalar field $\kappa$ on $O$, will be both flat and agree with $\bar{\nabla}$ on $\gamma$ but not agree with it everywhere.
    See also \citep[II.3]{iliev2006handbook} for a complete discussion of the freedom here.
    (A previous version of this footnote contained a typo.  We are grateful to James Read, Nic Teh, and Niels Linnemann for drawing it to our attention.)}
\end{cor}
\begin{proof}
Fix a spacetime $(M,g_{ab})$ and curve $\gamma:I\rightarrow M$.
Choose any point $p\in \gamma[I]$ and let $O$ be some neighborhood of $p$ on which there exists a flat metric $\bar{g}_{ab}$ that agrees with $g_{ab}$, to first order, on $\gamma[I]\cap O$.
Let $\varphi:O\rightarrow O$ be any diffeomorphism on $O$ such that $\varphi$ takes each point in  $\gamma[I]\cap O$ to itself, $\varphi_*$ acts as the identity on the tangent space of each point on $\gamma[I]\cap O$, and $\varphi$ is not the identity outside $\gamma[I]\cap O$.
Then we claim $\varphi^*(\bar{g}_{ab})$ will also be flat and agree with $g_{ab}$ to first order on $\gamma[I]\cap O$.
For flatness, observe that the pullback of a metric always preserves its Riemann tensor.
For agreement, note that there exists a smooth field $C^a{}_{bc}$ on $O$ such that for any vector field $\xi^a$, $(\varphi^*(\nabla)_a - \nabla_a)\xi^b = C^b{}_{an}\xi^n$.
Now choose any vector fields $\xi^a,\eta^a$ on $O$.
Then we have, for any point $p\in\gamma[I]\cap O$, $C^b{}_{an}\eta^a\xi^n = \eta^a(\varphi^*(\nabla)_a - \nabla_a)\xi^b=\varphi_*(\eta^a)\nabla_a\varphi_*(\xi^b) - \eta^a\nabla_a\xi^b=\mathbf{0}$, where the second equality follows from the definition of the pullback on derivative operators and the third equality follows from the fact that $\varphi_*$ acts as the identity on the tangent space at each point of $\gamma[I]\cap O$ by construction.
Since $\eta^a$ and $\xi^a$ were arbitrary, it follows that $C^a{}_{bc}=\mathbf{0}$ on $\gamma[I]\cap O$. Finally, we note that since $\varphi$ is not the identity outside $\gamma[I]$, in general $\varphi^*(g_{ab})\neq g_{ab}$ on $O$.
\end{proof}

Thus, the failure of uniqueness of normal coordinates corresponds not just to the fact that there are many coordinate systems adapted to a given flat metric, but \emph{also} that there are many distinct flat metrics that approximate a given (curved) metric along a curve.
It follows that while every spacetime is locally approximately flat, none is canonically so, as there are many flat metrics locally approximating any given metric at any given point.
This failure of uniqueness is obscured, in our view, by approaches that focus on particular construction procedures, or on the existence of certain coordinates, because not all such flat metrics (or normal coordinates) arise from a single construction procedure.

Thus we see that the sense of local flatness captured by Theorem 1 offers some immediate advantages, mostly of clarity and elegance, over the second coordinate chart interpretation, above.
How significant are these advantages?  On the one hand, the claims we have made thus far could be rephrased using Lorentz normal coordinates, and they would still follow.
In that sense we have added nothing.
Nonetheless, we suggest that thinking in terms of flat approximating metrics leads to a fruitfully different perspective on a number of foundational questions, both about spacetime geometry and matter dynamics.
The remainder of this paper, and its sequel, are devoted to exploring aspects of that perspective.
Ultimately, the value of this way of thinking will turn on how fruitful it is in these applications (and others).

In the next two sections, we will draw out some further consequences of Theorem \ref{thm:normal} and Corollary \ref{cor:normal}.

\section{Is Flatness Special?}\label{sec:triviality}

In section \ref{sec:right}, we presented two ways of expressing the sense in which any relativistic spacetime is locally approximately flat.
We also argued that these statements expressed the strongest sense of local approximate flatness available---at least insofar as one cannot achieve approximation to order greater than 1.
Now we turn to the question of how to best understand the significance of these results.
In particular, we wish to investigate the role of \emph{flatness} in claims about local flatness.
When one claims that spacetime is locally (approximately) flat, or Minkowskian, should we understand such claims as implying that Minkowski spacetime is distinguished in this regard?
Is Minkowski spacetime, or regions thereof, the only spacetime that locally approximates all others?

As a first remark, there is a sense in which Minkowski spacetime is distinguished from other spacetimes, in a way that makes the fact that every spacetime is locally approximately flat especially salient.
This is because flat spacetime is often a much more convenient setting for performing calculations and other analyses, and often useful constructions---such as Fourier decompositions, vector and tensor integration, and so on---are only generally defined in that context.\footnote{
    There are special cases that admit of generalizations to curves spacetimes, such as integration involving differential forms and partial Fourier decompositions in spacetimes admitting certain symmetries (e.g., stationarity), but these are not the general cases to which we refer here.
}
It is useful to be able to immediately extend such constructions to approximate versions near a point or curve.
For this reason, Theorem \ref{thm:normal} and Corollary \ref{cor:normal} are of considerable pragmatic value.

But one might still ask whether there is a deeper sense in which flat spacetime is distinguished in Theorem \ref{thm:normal}.  And the answer is ``no''.  This result is more naturally expressed using the resources of Corollary \ref{cor:normal}.  While Minkowski spacetime features in the statement of that result in the above section, it can in fact be replaced with an arbitrary spacetime:
\begin{thm}\label{thm:universal}
Given any spacetime $(M,g_{ab})$, embedded curve $\gamma: I \to M$, point $p \in \gamma[I]$, compact neighborhood $U$ of $p$, Riemannian metric $h_{ab}$ on $U$, real $\epsilon > 0$, spacetime $(M',g'_{ab})$, and point $p' \in M'$, there exist neighborhoods $O \ni p$ and $O' \ni p'$, an embedded curve $\gamma': I' \to M'$ with $p' \in \gamma'[I']$, and an $(h,1,\epsilon)$-isometry $\chi: O' \to O$ between $(O,g_{ab})$ and $(O',g'_{ab})$ satisfying $\chi \circ \gamma' = \gamma$ on $I'$ and $\chi^*(g_{ab})=g'_{ab}$ on $\gamma'[I']$.
\end{thm}
\begin{proof}
Pick any point $\bar{p}$ of Minkowski spacetime.
By Corollary \ref{cor:normal},  there exist neighborhoods $\hat{O} \ni p$ and $\bar{O} \ni \bar{p}$, an embedded curve $\hat{\gamma}: \hat{I} \to \Reals^4$ with $\bar{p} \in \hat{\gamma}[\hat{I}]$, and an $(h_{|\hat{O}},1,\epsilon/2)$-isometry $\hat{\chi}: \bar{O} \to \hat{O}$ between $(\hat{O},g_{|\hat{O}})$ and $(\bar{O} ,\eta_{|\bar{O} })$ satisfying $\hat{\chi} \circ \hat{\gamma} = \gamma$ on $\hat{I}$ and $\hat{\chi}{}^*(g)=\eta$ on $\hat{\gamma}[\hat{I}]$.
Now, there is a linear isomorphism $\psi: T_{\bar{p}}\Reals^4 \to T_{p'}M'$ that preserves the classification of vectors into timelike, null, and spacelike.
In addition, the exponential map $\exp_{\bar{p}}: T_{\bar{p}}\Reals^4 \to \Reals^4$ is a diffeomorphism, and so is the exponential map $\exp_{p'}: T_{p'}M' \to M'$ onto its image.
Thus we may define, for a sufficiently small interval domain containing $\hat{\gamma}^{-1}(\hat{p})$, the curve $\gamma' = \exp_{p'} \circ \psi \circ \exp^{-1}_{\bar{p}} \circ \hat{\gamma}$, and by definition $p' \in \gamma'[I']$.

Next, note that for any sufficiently small compact neighborhood of $p$, $h'_{ab} = (\exp_{p'} \circ \psi \circ\exp^{-1}_{\bar{p}} \circ \hat{\chi}^{-1})^*(h_{ab})$ is a Riemannian metric on a compact neighborhood of $p'$.
So, by corollary \ref{cor:normal}, there exist neighborhoods $\hat{O}' \ni p$ and $\bar{O}' \ni \bar{p}$, an embedded curve $\hat{\gamma}': \hat{I}' \to \Reals^4$ with $\bar{p} \in \hat{\gamma}'[\hat{I}']$, and an $(h'_{ab},1,\epsilon/2)$-isometry $\hat{\chi}': \bar{O}' \to \hat{O}'$ between $(\hat{O}',g_{|\hat{O}'})$ and $(\bar{O}' ,\eta_{|\bar{O}' })$ satisfying $\hat{\chi}' \circ \hat{\gamma}' = \gamma'$ on $\hat{I}'$ and $\hat{\chi}'{}^*(g'_{ab})=\eta_{ab}$ on $\hat{\gamma}'[\hat{I}']$.
In particular, $\hat{\gamma}'$ and $\hat{\gamma}$ coincide where they are both defined since $\hat{\chi}'$ can be chosen to coincide with $\exp_{p'} \circ \psi \circ \exp^{-1}_{\bar{p}}$ where they are both defined.

Now define $I' = \hat{I} \cap \hat{I}'$, $O' = \hat{\chi}'[\bar{O} \cap \bar{O}']$, $\phi = \hat{\chi} \circ \hat{\chi}'{}^{-1}_{|O'}$, and $O = \chi[O']$.
On $I'$, $\chi \circ \gamma' = (\hat{\chi} \circ \hat{\chi}'{}^{-1}) \circ (\hat{\chi}' \circ \hat{\gamma}') = \hat{\chi} \circ \hat{\gamma} = \gamma$.
On $\gamma'[I']$, $\hat{\chi}{}^*(g_{ab}) = \hat{\chi}'{}^*(g'_{ab})$, so at the same points, $\chi^*(g_{ab})= (\hat{\chi} \circ \hat{\chi}'{}^{-1})^*(g_{ab}) =g'_{ab}$.
Moreover, $$\epsilon/2 > d_{\hat{O}}(g,\hat{\chi}_*(\eta);h,1) = d_{\bar{O}}(\hat{\chi}^*(g),\eta;\hat{\chi}^*(h),1) \geq d_{\bar{O}\cap \bar{O}'}(\hat{\chi}^*(g),\eta;\hat{\chi}^*(h),1)$$ and \begin{align*}
    \epsilon/2 &> d_{\hat{O}'}(g',\hat{\chi}'_*(\eta);h',1) = d_{\bar{O}'}(\hat{\chi}'^*(g'),\eta;\hat{\chi}'^*(h'),1)\\
    &\geq d_{\bar{O} \cap \bar{O}'}(\hat{\chi}'^*(g'),\eta;\hat{\chi}'^*(h'),1) = d_{\bar{O} \cap \bar{O}'}(\hat{\chi}'^*(g'),\eta;\hat{\chi}^*(h),1).
    \end{align*}
So by the triangle inequality, $\epsilon > d_{\bar{O}\cap \bar{O}'}(\hat{\chi}^*(g),\hat{\chi}'^*(g');\hat{\chi}^*(h),1) = d_O(g,\chi_*(g');h,1)$, i.e., $\chi$ is an $(h,1,\epsilon)$-isometry.
\end{proof}

If we call any spacetime fulfilling the role of Minkowski spacetime in Corollary \ref{cor:normal} a \textit{universal locally approximating spacetime}, then Theorem \ref{thm:universal} shows that \textit{every} spacetime is a universal locally approximating spacetime.
For example, one could equally well take (anti-)de Sitter spacetime or Schwarzschild spacetime to play this role.\footnote{
    See \citet{wise2010macdowell} for an application of this idea using Cartan geometry to describe MacDowell-Mansouri gravity.}
So, it may be misleading to assert that ``free-falling observers see no effect of gravity in their immediate vicinity'' \citep[p.~11]{poisson2004relativist}; one might just as well say ``free-falling observers see the local effects of a large cosmological constant'' or ``free-falling observers see the local effects of being inside a rotating black hole''.  This is because, along any curve and, approximately, in a neighborhood of any segment of the curve, the geometrical features of \textit{all} spacetimes are indistinguishable to first order. This seems to be a general feature of metric theories of gravity, for none of these results require Einstein's equation.

All this said, the fact that other spacetimes are universal locally approximating does not imply that Minkowski spacetime is \emph{not}---and so one might ask whether there are \textit{other} reasons to think that Minkowski spacetime has a distinguished role to play (beyond its pragmatic advantages already noted).  One possible answer would return to an issue we raised previously, in section \ref{sec:wrong}: in some discussions of local (approximate) flatness, authors present particular constructions of normal coordinates, or flat approximating metrics, motivated by physical considerations.  For instance, as we noted above, Fermi normal coordinates along a timelike geodesic may be thought of as the coordinates that a certain kind of idealized inertial observer might assign to spacetime---the vectors of the associated frame might represent something like an ideal clock and rigid measuring rods.  The fact that these coordinates may be interpreted as standard Minkowski coordinates adapted to a particular flat metric that approximates the spacetime metric along the observer's worldline might be taken to give Minkowski spacetime a special status as a universal approximating spacetime.  In other words, the argument would go, it is not just that spacetime is locally approximately flat; it is that certain observers, under certain idealized circumstances and using certain prescribed procedures, would naturally construct a particular approximating metric, which happens to be Minkowskian---and not, say, (anti-) de Sitter.  Authors who invoke local (approximate) flatness to explain the success of special relativity may well have something like this argument in mind.

Perhaps this is true---though we emphasize that it is not clear how this argument really yields a special sense in which spacetime is locally approximately flat.
Moreover, this interpretation of the Fermi normal coordinate construction is not conceptually innocent, as perfectly rigid objects exist only under very special circumstances in general relativity, circumstances not fulfilled in most cases of interest (e.g., involving acceleration or geodesic deviation).
But even if we set that issue aside, it remains the case that this sort of argument purchases a special status for Minkowski spacetime at the cost of \textit{assuming} a special status for a particular coordinate construction procedure.
There are two aspects to this assumption.
First is the restriction to timelike geodesics; once this is relaxed, the resulting coordinates may not be Lorentz coordinates, as the Rindler coordinates generated by Fermi transport of a frame for a uniformly accelerating observer in Minkowski spacetime attest.
The second aspect is that while the Fermi normal coordinate construction is mathematically convenient, there is nothing physically unique about it.
Other coordinate construction procedures can be specified to generate local coordinates natural to any other spacetime.
Which one chooses, if any, depends on the pragmatics of representing or predicting quantities of interest.

\section{Local Approximate Poincar\'e Symmetry}\label{sec:symmetry}

There is another variant on the claim that spacetime is locally approximately Minkowskian that one sometimes sees in the literature, according to which spacetime is said to exhibit, locally and approximately, the symmetries of Minkowski spacetime: that is, that spacetime is \emph{locally approximately Poincar\'e invariant}.  In this penultimate section, we turn to discuss this claim in light of Theorems \ref{thm:normal} and \ref{thm:universal}.

In fact, several different notions of ``local approximate Poincar\'e invariance'' are to be found in the literature; here, we focus on one recently introduced by \citet{Fletcher2020}.\footnote{\citet{ReadBrownLehmkuhl2018} have also introduced notions of ``local Poincar\'e invariance'', but \citet{Fletcher2020} argues that the definition they give of local Poincar\'e invariance of a spacetime is unsatisfactory. (See also \citet{WeatherallDogmas}.)  The alternative definition that Fletcher proposes is intended to address somewhat different issues from those that concern Read et al., and it is not clear that it can play the role in their arguments that their own definition does.  But a full assessment of that question would take us too far afield, and so we postpone any discussion of the relationship between their arguments and the results here to future work.}  Fletcher defines $(h,k,\epsilon)$-approximate isometries (or symmetries) as we do here; he then considers smooth vector fields $\xi^a$ near a point $p$ whose one parameter families of diffeomorphisms, for sufficiently small parameter values, generate $(h,k,\epsilon)$-approximate symmetries on sufficiently small neighborhoods of $p$.  On his definition, a spacetime $(M,g_{ab})$ has \emph{local approximate Poincar\'e symmetry} to order $k$, relative to a Riemannian metric $h_{ab}$, near a point $p$ if: there exist ten smooth vector fields on a neighborhood of $P$ whose Lie derivatives with respect to one another satisfy the Poincar\'e algebra commutation relations; and for any $\epsilon>0$, the one parameter families of diffeomorphisms generated by those vector fields generate $(h,k,\epsilon)$-approximate symmetries on sufficiently small neighborhoods of $p$.  He then shows that every spacetime has local approximate Poincar\'e symmetry, to any order and relative to any metric $h_{ab}$, near every point.

In fact, though Fletcher does not emphasize this, on his definitions the following holds: \emph{any} smooth vector field $\xi^a$ defined near any point $p$ in any relativistic spacetime generates $(h,k,\epsilon)$-approximate symmetries, for all $k$ and $h_{ab}$.  This follows simply from the smoothness of all of the structures under consideration.\footnote{Indeed, a tensor field is smooth at a point $p$ iff every smooth vector field generates $(h,k,\epsilon)$-approximate symmetries, for all $k$ and $h_{ab}$, near $p$.} So the key move in his argument that every spacetime has local approximate Poincar\'e symmetry is to show that one can always find smooth vector fields, near any point, that satisfy the commutation relations of the Poincar\'e algebra.  He does so by an argument invoking the exponential map.  But the results above offer a different perspective on how these smooth vector fields arise as local representations of the Poincar\'e algebra.  In particular, given any spacetime $(M,g_{ab})$ and point $p\in M$, let $\bar{g}_{ab}$ be a flat metric on a neighborhood $O$ of $p$ that approximates $g_{ab}$ to first order, in the sense of Theorem \ref{thm:normal}.  Then this metric will have (local) Killing vector fields defined on $O$, which, since the metric is flat, will satisfy the Poincar\'e commutation relations.  These local Killing vector fields will generate local approximate Poincar\'e symmetries of $(M,g_{ab})$ near $p$.  This alternative construction is helpful, because it clarifies, again, that although every spacetime admits local approximate Poincar\'e symmetries, it does not do so \textit{uniquely}.  If one chose a different flat approximating metric near $p$, the local representation of the Poincar\'e algebra resulting from the present construction would be different.

In fact, a similar moral holds if one adopts a slightly stronger notion of when a smooth vector field generates an ``approximate local symmetry'' than  Fletcher explicitly endorses.  Let us say that a smooth vector field $\xi^a$ on a relativistic spacetime $(M,g_{ab})$ generates an approximate local symmetry$^*$ near a point $p\in M$ if $\mathcal{L}_{\xi} g_{ab} = \mathbf{0}$ at $p$.  This definition captures the idea that not only is the difference between $g_{ab}$ and the flow of $g_{ab}$ along $\xi^a$ bounded in sufficiently small neighborhoods, but that the derivative of $g_{ab}$ along $\xi^a$ also vanishes at $p$.  Call such a vector field $\xi^a$ an \emph{approximate local Killing vector field}.\footnote{This definition of an approximate local Killing vector field is apparently the same as one implicitly adopted by \citet[pp.~18--19]{Sus}.}  With this definition in hand, we can say that a spacetime admits \emph{local approximate Poincar\'e symmetries$^*$} if, on some neighborhood $O\ni p$, there exist ten approximate local Killing vector fields that (exactly) satisfy the commutation relations of the Poincar\'e algebra.

The argument sketched above establishes that every spacetime admits \emph{local approximate Poincar\'e symmetries$^*$}.  This is because, given any point $p\in M$ and any flat metric $\bar{g}_{ab}$ approximating $g_{ab}$ in the sense of Theorem \ref{thm:normal} at $p$, if $\xi^a$ is a Killing vector field for $\bar{g}_{ab}$ at $p$, then
\[
\mathbf{0}=\mathcal{L}_{\xi}\bar{g}_{ab} = 2\bar{\nabla}_{(a}\xi_{b)} = 2\nabla_{(a}\xi_{b)}=\mathcal{L}_{\xi}g_{ab}
\]
at $p$, where $\bar{\nabla}$ and $\nabla$ are the Levi-Civita derivative operators associated with $\bar{g}_{ab}$ and $g_{ab}$, respectively, and we have made use of the fact that, at $p$, $\bar{\nabla}=\nabla$.  Thus the exact Poincar\'e symmetries of the approximating metric $\bar{g}_{ab}$ gives rise to  approximate Poincar\'e symmetries$^*$ of the original metric, $g_{ab}$. Once again, the approximate local Poincar\'e symmetries$^*$ of a generic metric $g_{ab}$ will fail to be unique, in the sense that they will be realized relative to different representations of the Poincar\'e algebra near $p$, corresponding to the different flat metrics that approximate $g_{ab}$ to first order near $p$.

One can push this line of thought even further by making use of Theorem \ref{thm:universal}.  In particular, we have just seen that there is a relationship between the local (exact) symmetries of a flat metric that approximates a given metric $g_{ab}$ near a point $p$ and the local approximate Poincar\'e symmetries of $g_{ab}$.  But in light of Theorem \ref{thm:universal}, identical arguments show that the symmetries of \emph{any} relativistic spacetime can be implemented as local approximate symmetries near any point of any spacetime at all.\footnote{It is perhaps worth noting that not all local approximate symmetries, or symmetries$^*$ arise in this way.  For instance, every relativistic spacetime \emph{also} has local approximate Euclidean symmetries$^*$, even though no Riemannian metric can approximate a Lorentzian one in the sense of Theorem \ref{thm:normal}.  To see this, note that in normal coordinates near any point of any spacetime $(M,g_{ab})$, one can always construct a flat Riemannian metric whose Levi-Civita derivative operator is the coordinate derivative operator, and thus the Killing vector fields of that flat Riemannian manifold will be local approximate Killing vector fields of the spacetime metric, $g_{ab}$.}  In other words, while every spacetime admits local approximate Poincar\'e symmetries near any point, so, too, does every spacetime admit local approximate (anti-)de Sitter symmetries, local approximate Schwarzschild symmetries, and local approximate Kerr symmetries.

Taken together, these observations clarify just how weak the property of ``approximate local invariance under some spacetime symmetry group'' is.  We take this point to amplify remarks made by \citet{Fletcher2020} in his original discussion of approximate local Poincar\'e symmetries, where he argues that the existence of such symmetries has no logical relationship to any ``local symmetries'' (or other substantive properties) of matter equations. We also wish to emphasize an important distinction that is especially important when reasoning about approximate local symmetries: ``approximate local invariance under symmetry group $G$'' is importantly different from ``approximate local invariance under \emph{only} symmetry group $G$ in a single, specific representation''.  Every spacetime has local approximate Poincar\'e invariance in the first sense, but not the second; this means that any argument that relies on local approximate Poincar\'e invariance as a premise will presumably still go through if one substituted any of the myriad other local approximate symmetries groups of spacetime.\footnote{Compare with \citet[\S 5.2]{Sus}.}

\section{Interlude}\label{sec:conclusion}

In this Part, we have considered several possible interpretations of the claim that relativistic spacetimes are ``locally approximately Minkowskian'' or ``locally approximately flat''.  We argued that two possible interpretations were simply false and that a third was too weak to have substance.  On the fourth interpretation we offered---the second of two coordinate chart interpretations---the claim is true, but its significance was difficult to fully assess.  We then stated and proved Theorem \ref{thm:normal}, which captures a precise sense in which every relativistic spacetime is locally approximately flat.  This final interpretation says that every spacetime is locally approximately flat in the sense that near any point of any spacetime (or near sufficiently small segments of a curve), there exists a flat metric that coincides with the spacetime metric to first order at that point (or on that curve) and approximates it arbitrarily well, relative to a particular family of norms, near that point (or curve).

This final interpretation is closely related to the second coordinate chart interpretation just mentioned.  But recasting things as we did, in terms of the existence of an approximating flat metric, allowed us to clarify certain features of local flatness that do not appear to have been widely recognized before.  In particular, we showed, in section \ref{sec:right}, that the approximating flat metric is not unique.  In other words, while it is true that every spacetime is locally approximately flat, it is not canonically so---which means that one cannot unambiguously invoke ``the'' approximate Minkowski structure associated with a point or curve.  Particular physical constructions or idealized observational contexts might suggest particular choices of approximating flat metric, but these are privileged only relative to those further choices.

We also showed that although there is a sense in which every spacetime is locally approximately Minkowskian, Minkowski spacetime is not the unique universal approximating spacetime, and that, in fact, every spacetime locally approximates every other spacetime.  We used this result to cast doubt on the idea that the local approximate Minkowski character of spacetime carried great foundational (as opposed to pragmatic) significance.  The upshot of all of this is that while one can isolate a precise and accurate statement to the effect that spacetime is locally approximately Minkowskian, this statement is misleadingly specific given that local approximation is pervasive. Perhaps a better way of characterizing the situation is that, to first order, \textit{all} spacetimes with the same metric signature have a universal character, in the sense that they all locally approximate one another.\footnote{Beyond general relativity, this seems to be a general feature of metric theories of gravity, relativistic or not: the possibilities they allow are not distinguishable by their purely local features.}  It is only at second order and higher that differences in structure between different spacetimes can be seen in arbitrarily small neighborhoods of a point or curve.  That spacetimes cannot approximate one another arbitrarily well to second order is, of course, closely related to the fact that curvature is a tensor.

Finally, we used these results to show that although every spacetime admits local approximate Poincar\'e symmetries near any point, in the sense introduced by \citet{Fletcher2020} and in another, slightly stronger sense that we introduced here, there are in general infinitely many ways in which they do so, and so, again, one cannot unambiguously speak of ``the'' local approximate Poincar\'e symmetries of a spacetime.  Indeed, not only does every spacetime admit local approximate Poincar\'e symmetries in many distinct senses (i.e., under different representations), every spacetime also locally approximately exhibits many other symmetries.

In the next Part, we will further develop and apply these ideas in the context of a related set of claims, to the effect that general relativity is ``locally special relativistic'' because matter in general relativity behaves, locally, as if it were in flat spacetime.


\section*{Acknowledgments}
SCF acknowledges the support of a Marie Curie fellowship (FP7-MC-IIF-628533) during the early development of this project, and helpful feedback from audiences in T\"ubingen (3rd International Interdisciplinary Summer School), London (Sigma Club), Munich (MCMP), Vienna (Center for Quantum Science and Technology), Bucharest (Philosophy Dept.), Salzburg (Philosophy Dept.), and Dubrovnik (42nd Annual Philosophy of Science Conference) on an ancestral version entitled, ``On the Local Flatness of Spacetime.'' JOW: This material is partially based upon work produced for the project “New Directions in Philosophy of Cosmology”, funded by the John Templeton Foundation under grant number 61048, and partly upon work supported by the National Science Foundation under Grant No.~1331126.  We are grateful to David Malament and an anonymous reviewer for helpful comments on a previous draft, to to Thomas Barrett and JB Manchak for discussions of related material, and to Niels Linneman, James Read, and Nic Teh for drawing out attention to an important typo in a previous draft.

\singlespacing

\bibliographystyle{elsarticle-harv}
\bibliography{flatness}

\end{document}